\documentclass[sigconf]{acmart}
\AtBeginDocument{%
  \providecommand\BibTeX{{%
    \normalfont B\kern-0.5em{\scshape i\kern-0.25em b}\kern-0.8em\TeX}}}

\copyrightyear{2022}
\acmYear{2022}
\setcopyright{rightsretained}
\acmConference[ICPP '22]{51st International Conference on Parallel Processing}{August 29-September 1, 2022}{Bordeaux, France} \acmBooktitle{51st International Conference on Parallel Processing (ICPP '22), August 29-September 1, 2022, Bordeaux, France}\acmDOI{10.1145/3545008.3545025} \acmISBN{978-1-4503-9733-9/22/08}

\usepackage{natbib}
\usepackage{amsmath,amsfonts}
\usepackage{algorithmic}
\usepackage{graphicx}
\usepackage{textcomp}
\usepackage{xcolor}
\def\BibTeX{{\rm B\kern-.05em{\sc i\kern-.025em b}\kern-.08em
    T\kern-.1667em\lower.7ex\hbox{E}\kern-.125emX}}

\usepackage{booktabs}
\usepackage{centernot}
\usepackage{algorithm}
\usepackage{mathtools}
\usepackage{hyperref}
\usepackage{amsthm}

\usepackage{todonotes}
\newboolean{todo}
\setboolean{todo}{true} 

\newcommand{\remarkInternal}[4]{\ifthenelse{\boolean{todo}}{\todo[inline, color=#2, caption={2do}, #3]{\begin{minipage}{\textwidth-4pt}\emph{Remark #1:}\\#4\end{minipage}}}{}}

\DeclareMathOperator*{\argmin}{arg\,min}

\newcommand{\innermid}{\;\middle\lvert\;}

\newtheorem{theorem}{Theorem}
\newtheorem{proposition}{Proposition}

\allowdisplaybreaks

\begin{document}

\title{Learning Mean-Field Control for Delayed Information Load Balancing in Large Queuing Systems}

\author{Anam Tahir}
\authornote{These authors contributed equally to this research.}
\author{Kai Cui}
\authornotemark[1]
\author{Heinz Koeppl}
\email{{anam.tahir, kai.cui, heinz.koeppl}@tu-darmstadt.de}
\affiliation{%
  \institution{Department of Electrical Engineering and Information Technology}
  \city{Technische Universität Darmstadt}
  \country{Germany}
}



\begin{abstract}

Recent years have seen a great increase in the capacity and parallel processing power of data centers and cloud services. To fully utilize the said distributed systems, optimal load balancing for parallel queuing architectures must be realized. Existing state-of-the-art solutions fail to consider the effect of communication delays on the behaviour of very large systems with many clients. In this work, we consider a multi-agent load balancing system, with delayed information, consisting of many clients (load balancers) and many parallel queues. In order to obtain a tractable solution, we model this system as a mean-field control problem with enlarged state-action space in discrete time through exact discretization. Subsequently, we apply policy gradient reinforcement learning algorithms to find an optimal load balancing solution. Here, the discrete-time system model incorporates a synchronization delay under which the queue state information is synchronously broadcasted and updated at all clients. We then provide theoretical performance guarantees for our methodology in large systems. Finally, using experiments, we prove that our approach is not only scalable but also shows good performance when compared to the state-of-the-art power-of-d variant of the Join-the-Shortest-Queue (JSQ) and other policies in the presence of synchronization delays.
\end{abstract}

\ccsdesc[500]{Computing methodologies~Multi-agent reinforcement learning}
\ccsdesc[500]{Networks~Network resources allocation}

\keywords{load balancing; parallel systems; mean-field control; reinforcement learning}

\maketitle

\section{Introduction}
\label{sec: related work}

Load balancing in large queuing systems has been of great interest in the field of parallel processing and has yielded many successful distributed algorithms such as Join-the-Shortest-Queue (JSQ), Shortest-Expected-Delay (SED) \cite{winston1977optimality, selen2016steady, whitt1986deciding} and many others, see also \cite{van2018scalable} for a recent review. JSQ and SED have been designed for asynchronous systems with a central dispatcher (agent / client) assigning jobs (packets) to $M$ parallel servers (queues) under the assumption that the dispatcher can obtain instantaneous, accurate and synchronized information of the queue lengths at all times. In practice, both instant information and centralized dispatching are not realistic, especially if the number of queues $M \gg 1$ is large.


\begin{figure}
\center
\includegraphics[width=0.85\linewidth]{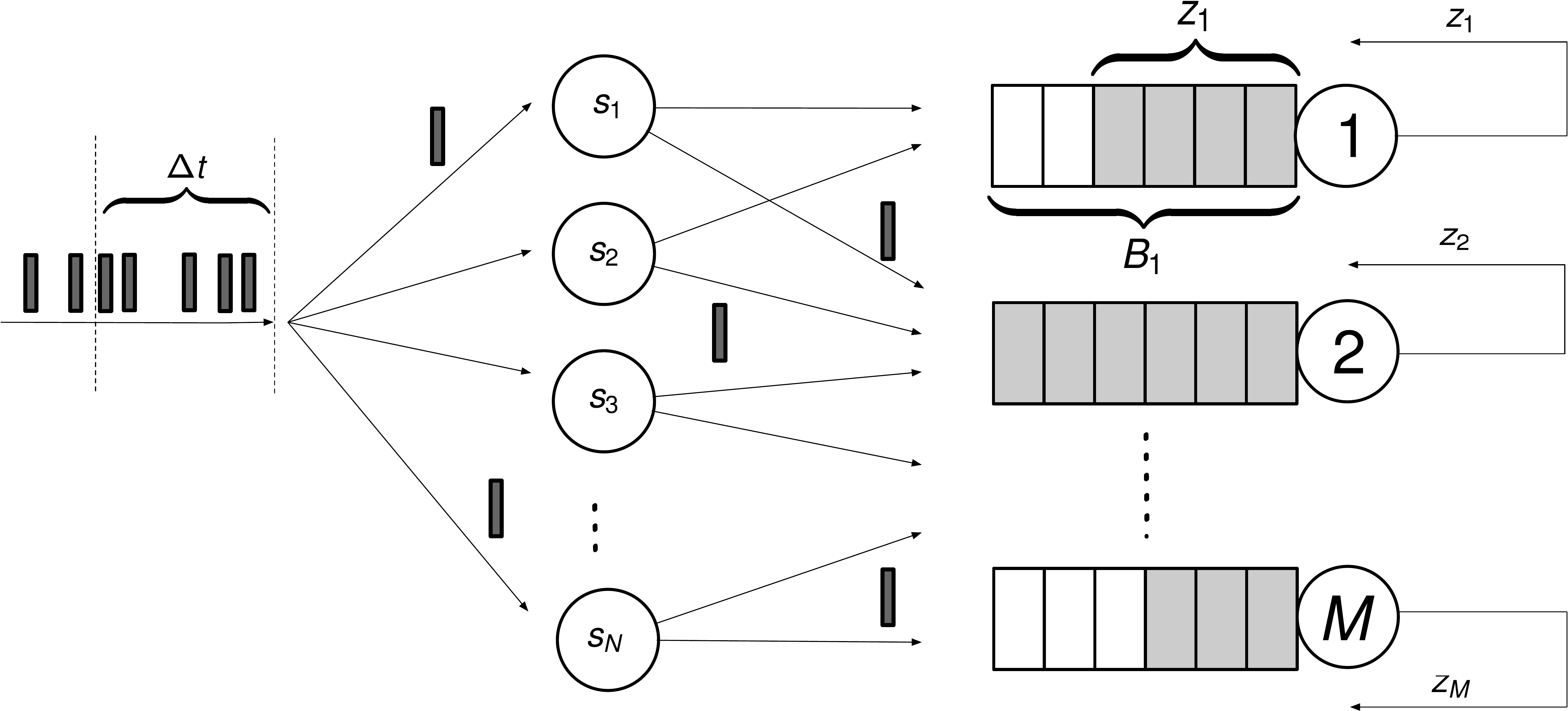}
    \caption{Our system model consists of $N$ clients and $M$ parallel servers. Jobs arriving in a certain time interval $\Delta t$ are assigned to the clients, which consequently assign them to one of a few sampled servers based on some policy. Arrows from each client indicate the $d=2$ servers randomly sampled by each client at the current epoch.} 
 \label{fig:system model}
\end{figure}


To remedy this scalability issue, the power-of-$d$ versions JSQ($d$) and SED($d$) of JSQ and SED \cite{mitzenmacher2001power} let the dispatcher sample only $d \le M$ out of $M$ servers randomly and then allocate the job to the sampled server with shortest expected processing time. However, JSQ($d$) and SED($d$) nonetheless assume instant and accurate information of the state of those $d$ servers, which remains unrealistic due to both the distributed nature of the system and computational overheads introducing latency. The problem is only exacerbated in a multiple client scenario where all clients access simultaneously. Hence, to model a more realistic system, it is of importance to take communication delays $\Delta t$ into account. In \cite{mitzenmacher2000useful}, it was shown that JSQ fails when $\Delta t > 0$ mainly due to a phenomenon known as ‘herd behaviour’: Multiple clients assigning jobs at the same time would consider the same subset of servers with few jobs, and thus all clients will end up assigning to the same servers. This eventually leads to higher response times and, in the case of finite queues, job drops. Though JSQ($d$) ameliorates this issue somewhat since it is highly unlikely for small $d$ and large $M$ that many clients will randomly choose the same servers, the technique nonetheless remains suboptimal under delayed information. Indeed, as $\Delta t \to \infty$, a completely random allocation to one of the servers becomes optimal  \cite{mitzenmacher2001power}. However, when the delay $\Delta t$ lies between $0$ and $\infty$, the optimal policy must lie in-between, which will be the main focus of this work.

In this paper, we shall consider a multi-agent system of $N$ clients and $M$ servers with $N \gg M \gg 1$ and communication delay. For scalability, each client samples $d$ of the $M$ servers uniformly at random using the power-of-$d$ method. The discretized system can be understood as a delayed periodic or synchronously updating system. Most importantly, as a result of delayed information, the number of agents will make a difference as opposed to the delay-free case, since each agent may see a different subset of information. In order to scale to a great number of clients and servers, we will apply mean-field theory, analogous to fluid limits $M \to \infty$, that is used to tractably model and assess systems with many queues. Fluid limits were used to study the performance of scheduling algorithms like JSQ and JSQ(d) in terms of sojourn time and average queue length \cite{mitzenmacher2001power, mukherjee2018universality, dawson2005balancing}. However, models including delayed information still remain an open problem \cite{lipshutz2019open}, in particular in the presence of many clients. One work with similar system model and synchronization delays is given in \cite{zhou2021asymptotically}, though they instead consider finitely many servers with infinite buffer sizes where the multiple clients use their local, asynchronous estimates of queue lengths to perform scheduling. 
This idea of using local client memory has also been proposed in \cite{van2019hyper, anselmi2020power}, however only for a single client.

More generally, the same tractability issue for large systems has led to the increasing popularity of general (competitive) mean-field games (MFG) \cite{huang2006large, lasry2007mean, saldi2018markov} and their cooperative counterpart of mean-field control (MFC) \cite{andersson2011maximum, bensoussan2013mean, arabneydi2014team, djete2019mckean, cui2021discrete}, wherein a system with large numbers of interchangeable and indistinguishable agents is converted into a system where one representative agent is interacting with the distribution (mean-field) of other agents. Here, there has been great recent focus on learning-based solution algorithms for MFGs \cite{guo2019learning, subramanian2019reinforcement, cui2021approximately, aggarwal2021machine} and MFC \cite{carmona2019model, gu2020mean, mondal2021approximation}. We will similarly apply the enlarged state-action space technique for MFCs (see e.g. \cite{gu2020mean}), its associated dynamic programming principle as well as reinforcement learning in order to find optimal load balancing policies for otherwise intractably large system. While reinforcement learning (RL) \cite{sutton2018reinforcement}, so-far has found great success e.g. in games \cite{mnih2015human, brown2019superhuman}, robotics \cite{kober2013reinforcement} or communication and queuing networks \cite{luong2019applications, aggarwal2021machine}, in the case of multiple agents, there still remain many challenges in multi-agent reinforcement learning (MARL) such as intractability for large numbers of agents \cite{zhang2021multi}. RL itself has long since been used in numerous works -- though not in the context of mean-field control -- to find an optimal load balancing policy. For examples, see \cite{winston1977optimality, stidham1993survey, krishnan1987joining, li2019overview} and references therein. The combination with mean-field control allows for tractable solution of very large load balancing systems and shall be the subject of our studies. We will similarly formulate a synchronous system model with delay by assuming $N \gg M \to \infty$, which will allow us to apply reinforcement learning to the otherwise difficult to solve optimal load balancing problem. Although our model shares similarities in concept to MFC, it does not immediately fit into the framework of conventional MFC, as we not only derive the discrete-time mean-field model starting from an underlying continuous-time dynamic, but at the same time take a double limit of infinitely many queues and agents. While, existing MFC frameworks typically focus only on the limit of infinitely many agents without external dynamics of non-agent-bound (queue) states. 

To summarize our contributions, (i) we consider a model not only with synchronous communication delay, but also under the limit of both many clients and many servers, stepping towards a general solution for the outstanding problem of scalable load balancing under delayed information \cite{zhou2021asymptotically}; (ii) we formulate the system as a mean-field control problem, introducing a decision hierarchy to obtain a standard Markov decision process amenable to standard solution techniques; (iii) we theoretically show the well-motivatedness of our limiting model by proving that the limiting system performance is reached with arbitrary precision in sufficiently large systems; and (iv) we apply reinforcement learning to solve the otherwise difficult-to-solve Markov decision process with continuous, high-dimensional action space, at a complexity independent of the number of clients $N$ and servers $M$. We find that, as the synchronization delay increases to an intermediate value, the choice of the shortest queues or fully random assignment becomes suboptimal and is outperformed by a learned policy. This policy can either be learned offline for a given system with known parameters, or applied online to learn optimal assignments in live systems. Our claims are supported both theoretically and experimentally and ablated for the case where our formal assumption $N \gg M$ is violated, giving us a good solution for large-scale load balancing systems with many clients and servers.

\section{Load Balancing With Delay} 
\label{sec:system model}
In this section, we will introduce the problem setting that will motivate our formulation. An overview of the considered load balancing system is given in Figure~\ref{fig:system model}.

We consider $N$ clients and $M$ servers, where each server has its own queue with limited buffer capacity. Jobs arrive randomly according to a Markov modulated Poisson process -- modelling e.g. changing load factors throughout a day -- with rate $\lambda_t M$ and are divided uniformly among clients, which will allocate the jobs to servers for processing. In accordance with the power-of-$d$ technique, clients shall randomly select $d$ out of $M$ queues and -- according to some policy to be optimized -- send their jobs to a selection of these $d$ queues, where $d \ll M$. On the queuing side of our system model, we have $M$ parallel and homogeneous servers in the system with service rates $\alpha$. The queues are finite with a maximum buffer capacity $B$ and the jobs in the queues are served in a first-in-first-out (FIFO) manner. Each server sends back its queue filling status, which is then used by the clients to make their decision for the next incoming jobs. The number of jobs that are currently in each queue together make up the state of the environment. Our goal is to \emph{minimize overall job drops} under decentralized decision-making by each client, e.g. like in edge computing scenarios.

We will assume that our system operates synchronously and broadcasts updates of sampled queue states to dispatchers only once every fixed time interval. Thus, in the following we will model our system at discrete decision epochs $\{0, \Delta t, 2 \cdot \Delta t, \ldots\}$ for some synchronization delay $\Delta t > 0$, after each of which the clients will sample $d$ new queues and keep this selection of $d$ queues for the entire duration of that decision epoch. Not only will this allow us to incorporate communication delays, but it will also lead to significantly less sampling of server states by the clients, as each client is only required to sample $d$ servers in every decision epoch. Another advantage of this approach is that the resulting discretized Markov decision process will allow us to apply powerful and well-established reinforcement learning algorithms, which to this date have been extensively developed for discrete-time models.





\subsection{Mathematical model}
\textbf{Notation.} \textit{Let $\mathcal S$ be a discrete space equipped with the discrete topology. Define by $\mathcal P(\mathcal S)$ the space of (Borel) probability measures on $\mathcal S$, equipped with the $l_1$-norm $\lVert \mu - \nu \rVert_1 = \sum_{s \in \mathcal S} \left| \mu(s) - \nu(s) \right|$. To keep notation simple, we denote the probability mass function of $\nu \in \mathcal P(S)$ by $\nu(\cdot)$. In the following, we denote random variables of the finite system with superscript $N,M$, of the infinite-agent version with superscript $M$ and of the limiting mean-field system without superscript.}

Formally, the $N$-agent $M$-queue system could be considered a multi-agent Markov Decision Process (MMDP) for $N, M \in \mathbb N$, i.e. the cooperative and fully observable case. See e.g. \cite{oliehoek2016concise} for a review of possible multi-agent problem formulations. In principle, one could even consider competitive or partially observed cases. However, the resulting limiting mean-field systems will be significantly more complex and thus remain outside of our scope. Instead, we will in the following consider a decentralized control setting where agents, due to the symmetry of our model, shall act depending on the current distribution of queue states.

Define $\mathcal Z \coloneqq \{0, \ldots, B\}$ as the finite queue state space, i.e. each server can contain at most $B$ jobs in its queue. The agent state space shall be denoted as $\mathcal X \coloneqq \{1, \ldots, M\}^d$, i.e. a selection of $d$ random queues. Although we could disallow repeated queue selections, it will make no difference in sufficiently large systems and adds unnecessary notational complexity. Finally, each agent can choose as an action its choice of one of $d$ randomly sampled accessible queues, i.e. the action space is defined as the $d$ possible queue choices $\mathcal U \coloneqq \{1, \ldots, d\}$. At any decision epoch $t = 0, 1, \ldots$, the states and actions of agents $i = 1, \ldots, N$, are random variables denoted by $x_t^{N,M,i} \equiv (x_{t,1}^{N,M,i}, \ldots, x_{t,d}^{N,M,i}) \in \mathcal X$ and $u_t^{N,M,i} \in \mathcal U$, and similarly the state of each queue $j=1, \ldots, M$ is denoted by $z_t^{N,M,j} \in \mathcal Z$ with $z_0^{N,M,j} \sim \nu_0 \in \mathcal P(\mathcal Z)$ from some initial distribution $\nu_0$. Additionally, $\lambda^{N,M}_t > 0$ -- the arrival rate parameter -- will be modulated as an independent discrete-time Markov chain with state space $\Lambda$, i.e.
\begin{align}
    \lambda^{N,M}_{t+1} \sim P_\lambda(\lambda^{N,M}_t)
\end{align}
for some arbitrary transition kernel $P_\lambda$. 

Due to symmetry of the problem, for sufficiently many agents, the information about each specific queue's state becomes irrelevant to the problem. Thus, we assume some common, shared policy of the form $\pi_t \colon \mathcal P(\mathcal Z) \times \mathcal Z^d \times \Lambda \to \mathcal P(\mathcal U)$ for all agents, acting on the current $\mathcal P(\mathcal Z)$-valued random empirical queue state distribution
\begin{align} \label{eq:empirical}
    \mathbb H_t^{N,M} \coloneqq \frac{1}{M} \sum_{j=1}^M \delta_{z_t^{N,M,j}}
\end{align}
with Dirac measure $\delta$, the sampled queue states, and the current arrival rate. In practice, we may also drop dependence on the current arrival rate and empirical distribution, or estimate e.g. the empirical queue state distribution by sampling a subset of random queues, though both will complicate the theoretical analysis of the limiting MFC problem, as it would not be possible to formulate the limiting system as a standard, fully-observed Markov decision process. 

The dynamics for each agent $i$ are thus given by
\begin{align} \label{eq:xsample}
    x_t^{N,M,i} &\sim \otimes_{k=1}^d \mathrm{Unif}(\{1,\ldots,M\}), \\
    u_t^{N,M,i} &\sim \pi_t \left( \mathbb H_t^{N,M}, (z_t^{N,M,x_{t,1}^{i}}, \ldots, z_t^{N,M,x_{t,d}^{i}}), \lambda^{N,M}_t \right), \label{esq:usample}
\end{align}
i.e. at each decision epoch, the agents decide to which of their $d$ randomly sampled, accessible queues they decide to send their jobs to. For simplicity of exposition, this choice of destination is deterministic, though in our experiments we shall allow randomization for each packet. As a result, starting with $z_0^{N,M,j} \sim \nu_0 \in \mathcal P(\mathcal Z)$  for each queue $j$ and some initial queue state distribution $\nu_0$, for any queue $j$, the next queue state $z_{t+1}^{N,M,j}$ is obtained from the previous state $z_{t}^{N,M,j}$ by simulating a $\mathcal Z$-valued continuous-time Markov chain for $\Delta t$ time units, beginning with $z_{t}^{N,M,j}$ and decrementing or incrementing by $1$ at departure rate $\alpha > 0$ and arrival rate 
\begin{align} \label{eq:arrival}
    \lambda^{N,M,j}_t = M \lambda^{N,M}_t \cdot \frac 1 N \sum_{i=1}^N \sum_{k=1}^d \mathbf 1_{x_{t,k}^{N,M,i} = j} \mathbf 1_{{u_{t}^{N,M,i} = k}}
\end{align}  
respectively, ignoring jumps above $B$ or below $0$. Any arrivals beyond $B$ are counted in the average number of dropped packets 
\begin{align}
    D^{N,M}_t = \frac{1}{M} \sum_{j=1}^M D^{N,M,j}_t
\end{align}
per queue $j$ during each decision epoch $t$, which will constitute our objective through the discounted infinite-horizon objective
\begin{align}
    J^{N,M}(\pi) = \mathbb E \left[ - \sum_{t=0}^{\infty} \gamma^t D^{N,M}_t \right]
\end{align}
to be maximized with discount factor $\gamma \in (0,1)$. 

Note that we can rewrite \eqref{eq:arrival} as
\begin{align} \label{eq:arrivalempnm}
    \lambda^{N,M,j}_t = M \lambda^{N,M}_t \int_{\mathcal X \times \mathcal U} \sum_{k=1}^d \mathbf 1_{x_k = j} \mathbf 1_{u = k} \, \mathbb G_t^{N,M}(\mathrm dx, \mathrm du)
\end{align}
with the $\mathcal P(\mathcal X \times \mathcal U)$-valued empirical agent state-action distribution
\begin{align}
    \mathbb G_t^{N,M} \coloneqq \frac{1}{N} \sum_{i=1}^N \delta_{x_{t}^{N,M,i}, u_{t}^{N,M,i}} .
\end{align}
Intuitively speaking, when $N \gg M \gg 1$, this empirical distribution becomes deterministic and we need not track each queue state, but only their distribution. Similarly, only the overall distribution of all agent choices will matter, leading to the prospective limiting mean-field model derived in the sequel.

\subsection{Infinite-agent limit}
In the infinite-agent limit where $N \to \infty$, we obtain a limiting control problem with random external states (queue states). Consider the evolution of the $\mathcal P(\mathcal Z)$-valued empirical queue state distribution
\begin{align} \label{eq:empirical-M}
    \mathbb H_t^{M} \coloneqq \frac{1}{M} \sum_{j=1}^M \delta_{z_t^{M,j}}
\end{align}
as $N \to \infty$. Conditional on the queue states and arrival rate, $(x_{t}^{M,i}, u_{t}^{M,i})_{i=1,\ldots,N}$ are i.i.d. Therefore, it will be sufficient to consider only the statistics of a representative agent. By the law of large numbers, we obtain the deterministic agent state distribution
\begin{align}
    \tilde \mu_t \coloneqq \otimes_{k=1}^d \mathrm{Unif}(\{1,\ldots,M\}) \in \mathcal P(\mathcal X)
\end{align}
of agents by \eqref{eq:xsample}. The $\mathcal P(\mathcal X \times \mathcal U)$-valued agent state distribution
\begin{align}
    \mathbb G_t^M \coloneqq \mathcal G_t^M(\tilde \mu_t, h_t)
\end{align}
thus depends on $h_t \coloneqq \pi_t(\mathbb H_t^{M}, \cdot, \lambda^{M}_t)$, where we define 
\begin{align}
    \mathcal G_t^M(\tilde \mu, h)(x,u) \coloneqq \tilde \mu(x) h(u \mid (z_t^{M,x_1}, \ldots, z_t^{M,x_d})).
\end{align}

We observe that this state-action distribution is sufficient for characterizing system behaviour: Conditional on fixed $\lambda^{M}_t$ and $\{z_t^{M,1}, \ldots, z_t^{M,M}\}$, the arrival rate in \eqref{eq:arrival} becomes
\begin{align} \label{eq:arrivalexp}
    \lambda^{M,j}_t &= M \lambda^{M}_t \mathbb E \left[ \sum_{k=1}^d \mathbf 1_{x_{t,k}^{M,1} = j} \mathbf 1_{u_{t}^{M,1} = k} \right] \\
    &= M \lambda^{M}_t \int_{\mathcal X \times \mathcal U} \sum_{k=1}^d \mathbf 1_{x_k = j} \mathbf 1_{u = k} \, \mathbb G_t^{M}(\mathrm dx, \mathrm du)
\end{align}
by the law of large numbers, similar to \eqref{eq:arrivalempnm}. In other words, the empirical agent state-action distribution $\mathbb G_t^{N,M}$ is replaced by the limiting distribution $\mathbb G_t^{M}$.
   
\subsection{Infinite-queue limit}
\label{subsec:3c}
Finally, we derive the mean-field model in the limit as $M \to \infty$, i.e. formally $N \gg M \gg 1$. The random queue states are now replaced by the queue state distribution denoted by $\nu_t \in \mathcal P(\mathcal Z)$. Therefore, each agent state $x_t^i \in \mathcal X$ is now also replaced by the anonymous queue state $\bar z_t^i \in \mathcal Z^d$ instead of the actual queue index. The queue state distribution deterministically induces the agent state distribution 
\begin{align} \label{eq:prodnu}
    \mu_t \coloneqq \otimes_{k=1}^d \nu_t \in \mathcal P(\mathcal Z^d)
\end{align}
by assigning the $d$-dimensional product measure $\mu_t(\bar z) = \Pi_{k=1}^d \nu_t(\bar z_k)$ for any $\bar z \equiv (\bar z_1, \ldots, \bar z_d) \in \mathcal Z^d$. For any decision rule $h_t = \pi_t(\nu_t, \cdot, \lambda_t)$, this agent state distribution induces a state-action distribution
\begin{align}
    \mathbb G_t \coloneqq \mu_t \otimes h_t \in \mathcal P(\mathcal Z^d \times \mathcal U).
\end{align}

Now consider the random amount of arriving packets $P \sim \mathrm{Pois}(M \lambda_t \Delta t)$ in a time slot $\Delta t$. Since $N \gg M$ implies $N \gg P$, the probability of any single agent receiving more than one packet is negligible. This implies that almost all packets' destination queues will be i.i.d. random variables. As a result, since packets arrive with rate $M \lambda_t$ and i.i.d. destinations, for any $z \in \mathcal Z$, packets will equivalently arrive with rate $M \lambda'_t(z)$ in queues with state $z \in \mathcal Z$ by Poisson thinning, where
\begin{align}
\lambda'_t(z) = \lambda_t \int_{\mathcal Z^d \times \mathcal U} \mathbf 1_{\bar z_u = z} \, \mathbb G_t(\mathrm d\bar z, \mathrm du) .
\end{align}

By symmetry, these packets arrive uniformly at random in any arbitrary specific fixed queue in state $z$. For any specific queue with state $z$, the probability of assigning such a packet to that queue is therefore $\frac{1}{M \nu_t(z)}$, which results in an equivalent queue packet arrival rate of
\begin{align} \label{eq:finalrate}
    \lambda_t(z) \coloneqq \frac{M \lambda'_t(z)}{M \nu_t(z)} = \frac{\lambda'_t(z)}{\nu_t(z)} .
\end{align}

The informal derivation until now will be motivated more rigorously in Section~\ref{sec:theo} and numerically in Section~\ref{sec:experiments}.

\subsection{Exact discretization}
\label{subsec:exact_discretization}
The final step is to formulate a discrete-time optimal control problem from the delayed, synchronous system that allows for application of standard optimal control techniques such as reinforcement learning. To discretize the mean-field system exactly at times $\{0, \Delta t, 2 \cdot \Delta t, \ldots\}$, we generate the master equations for the evolution of a single queue's state over time between each of the discretization time points. The procedure is done analogously for the pre-limit systems. Consider a queue in state $z \in \mathcal Z$ at the beginning of a decision epoch $t$. Then, for any $h_t$, we define a $\mathcal Z$-valued continuous-time Markov chain $y$ through $y(0) = z$ and formulate its Kolmogorov forward equations
\begin{align}
    \dot {\mathbf P}^z = \mathbf Q^z \mathbf P^z, \quad \mathbf P^z(0) = \mathbf e_z
\end{align}
for the vector of queue state probabilities $\mathbf P^z(\tau) \in [0,1]^{\mathcal Z}$ at times $\tau \in [0, \Delta t]$ with
\begin{align}
    P_{z'}^z(\tau) \equiv \mathbb P(y(\tau) = {z'}), \quad \forall {z'} \in \mathcal Z
\end{align}
and the transposed transition rate matrix $\mathbf Q^z \coloneqq \mathbf Q(\nu_t, z) \in \mathbb R^{\mathcal Z \times \mathcal Z}$ where $\mathbf Q(\nu, z)$ is defined by
\begin{align}
    \mathbf Q(\nu, z)_{i,i-1} = \lambda_t(\nu, z) \coloneqq \frac{1}{\nu(z)}\lambda_t \int_{\mathcal Z^d \times \mathcal U} \mathbf 1_{\bar z_u = z} \, (\nu \otimes h_t)(\mathrm d\bar z, \mathrm du)
\end{align}
in accordance with \eqref{eq:prodnu} - \eqref{eq:finalrate}, $\mathbf Q(\lambda, z)_{i-1,i} = \alpha(z)$ for $i=1,\ldots,B$, $\mathbf Q(\lambda, z)_{i,i} = -\sum_j \mathbf Q(\lambda, z)_{j,i}$ for $i=0,\ldots,B$, and zero otherwise. Here, $\mathbf e_z$ denotes the $z$-unit vector.

Therefore, from the fraction $\nu_t(z)$ of queues in state $z \in \mathcal Z$ at time $t$, we will deterministically have the resulting fraction
\begin{align}
    \nu_{z,{z'}} = \nu_t(z) P^z_{z'}(\Delta t)
\end{align} 
of queues with state $z \in \mathcal Z$ in resulting state ${z'} \in \mathcal Z$ at the end of the decision epoch $\Delta t$. In total, we therefore have
\begin{align} 
    \nu_{t+1}({z'}) = \sum_{{z} \in \mathcal Z} \nu_{z,{z'}} = \sum_{{z} \in \mathcal Z} \nu_t({z}) P^{z}_{z'}(\Delta t), \quad \forall {z'} \in \mathcal Z . \label{eq:nuupdate}
 \end{align}

Computing the expected packet drops $D^z_t$ per queue with state $z \in \mathcal Z$ is done analogously by
\begin{align}
    \dot D^z_t = \lambda_t(z) P^z_B, \quad D^z_t(0) = 0
\end{align}
resulting in a per-queue average packet loss of
\begin{align} \label{eq:lossperq}
    D_t = \sum_{z \in \mathcal Z} \nu_t(z) D^z_t(\Delta t) .
\end{align}

For exact computation of the terms in \eqref{eq:nuupdate} - \eqref{eq:lossperq}, observe that we have the linear matrix differential equation
\begin{align}
    \begin{bmatrix}
        \dot {\mathbf P}^z \\
        \dot D^z_t
    \end{bmatrix}
    = 
    \underbrace{\begin{bmatrix}
        \mathbf Q^z & 0 \\
        \lambda_t(\nu_t, z) \cdot \mathbf e_B^T & 0
    \end{bmatrix}}_{\bar{\mathbf Q}^z \equiv \bar{\mathbf Q}(\nu_t, z)}
    \cdot 
    \begin{bmatrix}
        \mathbf P^z \\
        D^z_t
    \end{bmatrix}
\end{align}
where we define the extended rate matrices $\bar{\mathbf Q}(\nu_t, z)$ analogously to $\mathbf Q(\nu_t, z)$, and thus obtain exact discretization by
\begin{align}
    \begin{bmatrix}
        {\mathbf P}^z(\Delta t)  \\
        D^z_t(\Delta t)
    \end{bmatrix}
    =
    \exp{(\bar{\mathbf Q} \Delta t)} 
    \cdot 
    \begin{bmatrix}
        \mathbf e_z  \\
        0
    \end{bmatrix}
\end{align}
where $\exp(\cdot)$ denotes the matrix exponential.

\subsection{Upper-level decision process}
We can now obtain a Markov decision process (MDP) \cite{puterman2014markov} with state space $\mathcal P(\mathcal Z) \times \Lambda$ and action space $\mathcal H \coloneqq \{h \colon \mathcal Z^d \to \mathcal P(\mathcal U)\}$, since we have states $(\lambda_t, \nu_t)$ and actions $h_t$ following dynamics
\begin{align}
    (\lambda_{t+1}, \nu_{t+1}) &\sim P_\lambda(\lambda_t) \otimes \delta_{T_\nu(\nu_t, \lambda_t, h_t)} \\
    h_t &= \tilde \pi_t(\nu_t, \lambda_t)
\end{align}
where the transition function $T_\nu$ deterministically maps to $\nu_{t+1}$ according to \eqref{eq:nuupdate}, and the actions are given by a deterministic `upper-level' policy $\tilde \pi = \{ \tilde \pi_t \}_{t \geq 0}$, where $\tilde \pi_t \colon \mathcal P(\mathcal Z) \times \Lambda \to \mathcal H$. Here, the randomness of the system stems from the random packet arrival rate $\lambda_t$. Finally, by \eqref{eq:lossperq}, the objective becomes
\begin{align}
    J(\tilde \pi) = \mathbb E \left[ - \sum_{t=0}^{\infty} \gamma^t D_t \right] .
\label{eq:objective_finite}
\end{align}


\begin{figure}
    \center
    \includegraphics[width=0.8\linewidth]{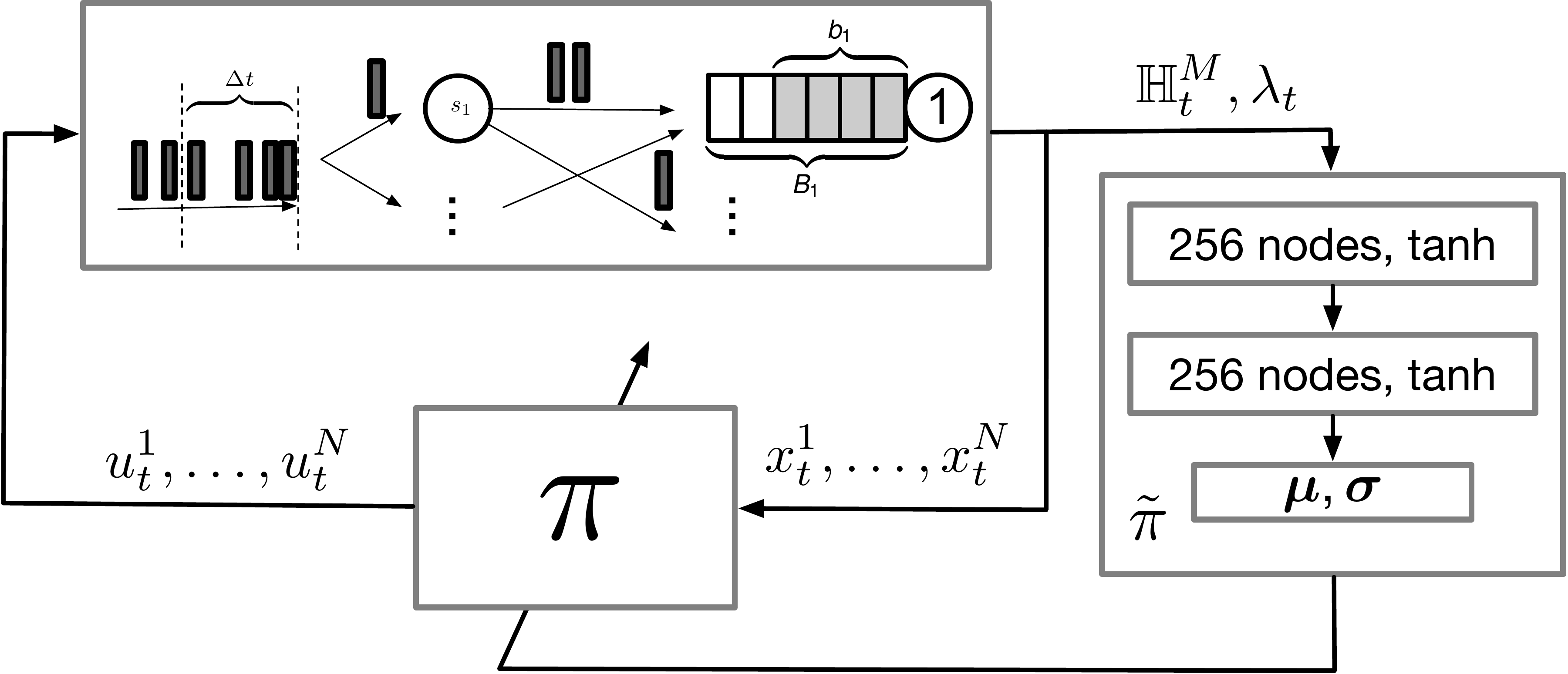}
    \caption{A schematic overview of the application of the upper-level mean-field control policy to the finite-client finite-server system. The upper-level policy $\tilde \pi$ returns a lower-level policy $\pi$ for a given distribution of server states $\mathbb H_t^M$ and current arrival rate $\lambda_t$. The lower-level policy is then applied separately to each agent state $x_t^i$ to obtain an action $u_t^i$.}
    \label{fig:mftonm}
\end{figure}


The application of $\tilde \pi$ to the $N$-agent, $M$-queue case is visualized in Figure~\ref{fig:mftonm}, i.e. each of the agents $i=1,\ldots,N$ first computes the decision rule $h_t = \tilde \pi_t(\mathbb H_t^M, \lambda_t)$ according to the upper-level policy, and then samples its action $u_t^i \sim h_t(x_t^i)$.

For the obtained MDPs, since the expected cost function and the dynamics are continuous in the states and actions of the MFC MDP, it is known that the typical dynamic programming principle (i.e. Bellman equation) holds, and an optimal stationary deterministic policy will exist.

\begin{proposition}[\cite{hernandez2012discrete}, Theorem 4.2.3]
There exists a stationary deterministic optimal policy $\tilde \pi$ that maximizes $J(\tilde \pi)$.
\end{proposition}

To find such a deterministic policy, an exact, closed-form solution is difficult due to the complexity of the associated transition model and continuous state and action spaces. Instead, we shall in the following employ well-established reinforcement learning techniques by exploring over stochastic policies $\tilde \pi_t \colon \mathcal P(\mathcal Z) \times \Lambda \to \mathcal P(\mathcal H)$, with the random decision rules $h_t \sim \tilde \pi_t(\nu_t, \lambda_t)$ as actions of the MFC MDP, to find the desired optimal stationary deterministic policy. 

It should be noted that in this section we have presented a system which has finite capacity queues with homogeneous servers, though this model can be extended to heterogeneous servers and infinite capacity queues, which we omit for space reasons.

\section{Theoretical Analysis}
\label{sec:theo}

Although our formulated mean-field model is intuitively a good approximation of the finite system, in this section we shall make this connection rigorous. Note that our model does not immediately fit into standard MFC frameworks introduced in \cite{gu2020mean, mondal2021approximation}, since we perform a double limit argument and continuous-to-discrete-time modelling. To verify the mean-field model, we shall show that performance in the finite system becomes arbitrarily close to the performance in the MFC system as long as the system is sufficiently large. Quantifying the error convergence rate more precisely is left to future work. For the following theoretical analysis, we shall consider the sequence of arrival rates $(\lambda_1, \lambda_2, \ldots)$ given a priori by conditioning on them, i.e. non-random $\lambda^{N,M}_t = \lambda^{M}_t = \lambda_t$.

\begin{theorem}
The performance of the $N,M$ system converges to the performance of the mean-field system under any stationary deterministic policy $\hat \pi$ as the system size becomes sufficiently large, i.e. for any $\varepsilon > 0$ there exists $N',M'(N') \in \mathbb N$ such that
\begin{align*}
    &\left| J(\hat \pi) - J^{N,M}(\hat \pi) \right| < \varepsilon
\end{align*}
for all $N > N', M > M'(N')$.
\end{theorem}
\begin{proof}
We will analyze
\begin{align*}
    \left| J(\hat \pi) - J^{N,M}(\hat \pi) \right| &\leq \left| J(\hat \pi) - J^{M}(\hat \pi) \right| + \left| J^{M}(\hat \pi) - J^{N,M}(\hat \pi) \right| \\
    &\leq \sum_{t=0}^{\infty} \gamma^t \left( \left| \mathbb E \left[ D_t - D^{M}_t \right] \right| + \left| \mathbb E \left[ D^{M}_t - D^{N,M}_t \right] \right| \right)
\end{align*}
where $D^{M}_t$ denotes the random loss of packets in the infinite-agent finite-queue system.

For the first term, consider $M \to \infty$ and observe that
\begin{align*}
    \mathbb E \left[ D_t \right] &= 
    \mathbb E \left[ \int \left( \exp{(\bar{\mathbf Q}(\nu_t, z) \Delta t)} 
    \cdot 
    \begin{bmatrix}
        \mathbf e_{z}  \\
        0
    \end{bmatrix} \right)_{B+1} \nu_t(\mathrm dz) \right], \\
    \mathbb E \left[ D^{M}_t \right] &= 
    \mathbb E \left[ \frac 1 M \sum_{j} 
    \left( \exp{(\bar{\mathbf Q}^{M,j} \Delta t)} 
    \cdot 
    \begin{bmatrix}
        \mathbf e_{z^{M,j}_t}  \\
        0
    \end{bmatrix} \right)_{B+1} \right] \\
    &= 
    \mathbb E \left[ \int
    \left( \exp{(\bar{\mathbf Q}(\mathbb H_t^{M}, z) \Delta t)} 
    \cdot 
    \begin{bmatrix}
        \mathbf e_{z}  \\
        0
    \end{bmatrix} \right)_{B+1} \mathbb H_t^{M}(\mathrm dz) \right],
\end{align*}
with the rate matrices $\bar{\mathbf Q}^{M,j}$ of the infinite-agent finite-queue system, where the last equality follows since the rates in the $M$-queue case for each queue $j$ are indeed given by 
\begin{align*}
    &\lambda^{M,j}_t = M \lambda_t \int_{\mathcal X \times \mathcal U} \sum_{k=1}^d \mathbf 1_{x_k = j \wedge u = k} \, \mathbb G_t^{M}(\mathrm dx, \mathrm du) \\
    &\quad = \lambda_t \sum_{k=1}^d \sum_{x \in \mathcal X} \sum_{u \in \mathcal U} \mathbf 1_{x_k = j \wedge u = k} \frac 1 {M^{d-1}} h_t(u \mid z_t^{M,x_1}, \ldots, z_t^{M,x_d}) \\
    &\quad = \lambda_t \sum_{k=1}^d \sum_{x_k \in \{1, \ldots, M\}} \sum_{x_{-k} \in \{1, \ldots, M\}^{d-1}} \\
    &\qquad\qquad \sum_{u \in \mathcal U} \mathbf 1_{x_k = j \wedge u = k} \frac 1 {M^{d-1}} h_t(u \mid z_t^{M,x_1}, \ldots, z_t^{M,x_d}) \\
    &\quad = \lambda_t \sum_{k=1}^d \sum_{x_k \in \{1, \ldots, M\}} \sum_{x_{-k} \in \{1, \ldots, M\}^{d-1}} \sum_{u \in \mathcal U} \sum_{\bar z_k \in \mathcal Z} \\
    &\qquad\qquad \sum_{\bar z_{-k} \in \mathcal Z^{d-1}} \mathbf 1_{x_k = j \wedge u = k} \frac 1 {M^{d-1}} h_t(u \mid (\bar z_k, \bar z_{-k})) \mathbf 1_{\bigwedge_{i=1}^d z_t^{M,x_i} = \bar z_i} \\
    &\quad = \lambda_t \sum_{k=1}^d \sum_{\bar z_k \in \mathcal Z} \sum_{\bar z_{-k} \in \mathcal Z^{d-1}} \sum_{u \in \mathcal U} \mathbf 1_{\bar z_k = z^{M,j}_t \wedge u = k} \\
    &\qquad\qquad \cdot \underbrace{\frac {\sum_{x_{-k} \in \{1, \ldots, M\}^{d-1}} \mathbf 1_{\bigwedge_{i \neq k} z_t^{M,x_i} = \bar z_i}} {M^{d-1}}}_{\prod_{i \neq k} \mathbb H^M_t(\bar z_i)} h(u \mid (\bar z_k, \bar z_{-k})) \\
    &\quad = \lambda_t \sum_{k=1}^d \sum_{\bar z \in \mathcal Z^d} \sum_{u \in \mathcal U} \mathbf 1_{\bar z_k = z^{M,j}_t \wedge u = k} \prod_{i \neq k} \mathbb H^M_t(\bar z_i) h_t(u \mid \bar z) \\
    &\quad = \lambda_t \sum_{\bar z \in \mathcal Z^d} \sum_{u \in \mathcal U} \mathbf 1_{\bar z_u = z^{M,j}_t} \prod_{i \neq u} \mathbb H^M_t(\bar z_i) h_t(u \mid \bar z) \\
    &\quad = \frac{\lambda_t \int_{\mathcal Z^d \times \mathcal U} \mathbf 1_{\bar z_u = z^{M,j}_t} \, (\mathbb H^M_t \otimes h_t)(\mathrm d\bar z, \mathrm du)}{\mathbb H^M_t(z^{M,j}_t)} = \lambda_t(\mathbb H^M_t, z^{M,j}_t)
\end{align*}
where the indices $-k$ denote all dimensions other than $k$. 

Therefore, as long as $\mathbb H_t^{M} \xrightarrow{d} \nu_t$ (convergence in distribution), we find $\mathbb E \left[ D_t - D^{M}_t \right] \to 0$ by the continuous mapping theorem. In particular, this holds true if $\mathbb H_t^{M} \xrightarrow{p} \nu_t$, i.e. for any $\delta > 0$ as $M \to \infty$,
\begin{align*}
    \mathbb P \left( \left\Vert \mathbb H_t^{M} - \nu_t \right\Vert > \delta \right) \to 0.
\end{align*}

We show this by induction: At $t=0$ the statement holds by the law of large numbers. Now assume that the statement holds for $t$, then for $t+1$ we first show that for any $\varepsilon, \delta > 0$ there exists $M', \delta' > 0$ such that for all $M > M'$ we have
\begin{align*}
    \mathbb P \left( \left\Vert \mathbb H_{t+1}^{M} - \nu_{t+1} \right\Vert > \delta \innermid \left\Vert \mathbb H_t^{M} - \nu_t \right\Vert \leq \delta' \right) < \varepsilon.
\end{align*}
Note that
\begin{align*}
    &\mathbb P \left( \left\Vert \mathbb H_{t+1}^{M} - \nu_{t+1} \right\Vert > \delta \innermid \left\Vert \mathbb H_t^{M} - \nu_t \right\Vert \leq \delta' \right) \\
    &\leq \sum_{z \in \mathcal Z} \mathbb P \left( \left| \mathbb H_{t+1}^{M}(z) - \nu_{t+1}(z) \right| > \delta \innermid \left\Vert \mathbb H_t^{M} - \nu_t \right\Vert \leq \delta' \right) \\
    &\leq \sum_{z \in \mathcal Z} \mathbb P \left( \left| \mathbb H_{t+1}^{M}(z) - \mathbb E \left[ \mathbb H_{t+1}^{M}(z) \innermid \mathbb H_{t}^{M} \right] \right| > \frac \delta 2 \innermid \left\Vert \mathbb H_{t}^{M} - \nu_t \right\Vert \leq \delta' \right) \\
    &\quad + \sum_{z \in \mathcal Z} \mathbb P \left( \left| \mathbb E \left[ \mathbb H_{t+1}^{M}(z) \innermid \mathbb H_{t}^{M} \right] - \nu_{t+1}(z) \right| > \frac \delta 2 \innermid \left\Vert \mathbb H_t^{M} - \nu_t \right\Vert \leq \delta' \right)
\end{align*}
and we shall bound the former term as follows: Define
\begin{align*}
    \Delta_{z_{t+1}^{M,j} \mid z_{t}^{M,j}} f \coloneqq f(z_{t+1}^{M,j}) - \mathbb E \left[ f(z_{t+1}^{M,j}) \innermid f(z_{t}^{M,j}) \right]
\end{align*}
and let $f \colon \mathcal Z \to \mathbb R$, then we have
\begin{align*}
    &\mathbb P \left( \left| \mathbb H_{t+1}^{M}(f) - \mathbb E \left[ \mathbb H_{t+1}^{M}(f) \innermid \mathbb H_{t}^{M} \right] \right| > \frac \delta 2 \innermid \left\Vert \mathbb H_{t}^{M} - \nu_t \right\Vert \leq \delta' \right) \\
    &\quad = \mathbb P \left( \left| \frac 1 M \sum_{j=1}^M \Delta_{z_{t+1}^{M,j} \mid z_{t}^{M,j}} f \right| > \frac \delta 2 \innermid \left\Vert \mathbb H_{t}^{M} - \nu_t \right\Vert \leq \delta' \right) \\
    &\quad \leq \frac 4 {\delta^2} \mathbb E \left[ \left( \frac 1 M \sum_{j=1}^M \left( \Delta_{z_{t+1}^{M,j} \mid z_{t}^{M,j}} f \right) \right)^2 \innermid \left\Vert \mathbb H_{t}^{M} - \nu_t \right\Vert \leq \delta' \right] \\
    &\quad = \frac 4 {\delta^2 M^2} \sum_{j=1}^M \mathbb E \left[ \left( \Delta_{z_{t+1}^{M,j} \mid z_{t}^{M,j}} f \right)^2 \innermid \left\Vert \mathbb H_{t}^{M} - \nu_t \right\Vert \leq \delta' \right] \\
    &\quad \leq \frac {16 \max_z f(z)^2} {\delta^2 M} \to 0
\end{align*}
as $M \to \infty$ by conditional independence of $(z^{M,1}_{t+1}, \ldots, z^{M,M}_{t+1})$ given $z^{M}_{t} = (z^{M,1}_{t}, \ldots, z^{M,M}_{t})$, the Chebyshev inequality and tower property. In particular, this holds for $f_z \equiv \mathbf 1_{\{z\}}$, $z \in \mathcal Z$. Therefore,
\begin{align*}
    \sum_{z \in \mathcal Z} \mathbb P \left( \left| \mathbb H_{t+1}^{M}(z) - \mathbb E \left[ \mathbb H_{t+1}^{M}(z) \innermid \mathbb H_{t}^{M} \right] \right| > \frac \delta 2 \innermid \left\Vert \mathbb H_{t}^{M} - \nu_t \right\Vert \leq \delta' \right) \to 0
\end{align*}
as $M \to \infty$. For the latter term, note that analogously
\begin{align*}
    &\left| \mathbb E \left[ \mathbb H_{t+1}^{M}(f) \innermid \mathbb H_{t}^{M} \right] - \nu_{t+1}(f) \right| \\
    &\quad \leq \left| \sum_{z \in \mathcal Z} f(z) \sum_{z' \in \mathcal Z} \left( \mathbb H_{t}^{M}(z) - \nu_t(z) \right) 
    \right.\\&\hspace{3cm}\left.
    \cdot \left( \exp{(\bar{\mathbf Q}(\mathbb H_t^{M}, z') \Delta t)} 
    \cdot 
    \begin{bmatrix}
        \mathbf e_{z'}  \\
        0
    \end{bmatrix} \right)_{z} \right| \\
    &\qquad + \left| \sum_{z \in \mathcal Z} f(z) \sum_{z' \in \mathcal Z} \nu_t(z) \cdot \left( \exp{(\bar{\mathbf Q}(\mathbb H_t^{M}, z') \Delta t)} 
    \cdot 
    \begin{bmatrix}
        \mathbf e_{z'}  \\
        0
    \end{bmatrix} 
    \right.\right.\\&\hspace{3cm}\left.\left.
    - \exp{(\bar{\mathbf Q}(\nu_t, z') \Delta t)} 
    \cdot 
    \begin{bmatrix}
        \mathbf e_{z'}  \\
        0
    \end{bmatrix} \right)_{z} \right|
\end{align*}
and by boundedness ($\lambda_t(\nu, z) \leq d \lambda_t$) and continuity in $\mathbb H_{t}^{M}, \nu_t$, for any $\varepsilon > 0$ there exists $\delta' > 0$ such that $\left\Vert \mathbb H_{t}^{M} - \nu_t \right\Vert \leq \delta'$ implies $\left| \mathbb E \left[ \mathbb H_{t+1}^{M}(f) \innermid \mathbb H_{t}^{M} \right] - \nu_{t+1}(f) \right| < \varepsilon$. As a result, by the law of total probability
\begin{align*}
    &\mathbb P \left( \left\Vert \mathbb H_{t+1}^{M} - \nu_{t+1} \right\Vert > \delta \right) \\
    &= \mathbb P \left( \left\Vert \mathbb H_{t+1}^{M} - \nu_{t+1} \right\Vert > \delta \innermid \left\Vert \mathbb H_t^{M} - \nu_t \right\Vert \leq \delta' \right) \cdot \mathbb P \left( \left\Vert \mathbb H_t^{M} - \nu_t \right\Vert \leq \delta' \right) \\
    &\quad + \mathbb P \left( \left\Vert \mathbb H_{t+1}^{M} - \nu_{t+1} \right\Vert > \delta \innermid \left\Vert \mathbb H_t^{M} - \nu_t \right\Vert > \delta' \right) \cdot \mathbb P \left( \left\Vert \mathbb H_t^{M} - \nu_t \right\Vert > \delta' \right) \\
    &\leq \mathbb P \left( \left\Vert \mathbb H_{t+1}^{M} - \nu_{t+1} \right\Vert > \delta \innermid \left\Vert \mathbb H_t^{M} - \nu_t \right\Vert \leq \delta' \right) + \mathbb P \left( \left\Vert \mathbb H_t^{M} - \nu_t \right\Vert > \delta' \right) \\
    &\to 0
\end{align*}
since we can choose $M', \delta'$ according to the former analysis and the induction assumption, completing the induction step. It then follows at all times $t$ by the continuous mapping theorem that
\begin{align*}
    \mathbb E \left[ D_t - D^{M}_t \right] \to 0.
\end{align*}

For the second term, fix $M$ and let $N \to \infty$. We find that
\begin{align*}
    \mathbb E \left[ D^{M}_t \right] &= 
    \frac 1 M \sum_{j} 
    \mathbb E \left[ \left( \exp{(\bar{\mathbf Q}^{M,j} \Delta t)} 
    \cdot \begin{bmatrix}
        \mathbf e_{z^{M,j}_t}  \\
        0
    \end{bmatrix} \right)_{B+1} \right], \\
    \mathbb E \left[ D^{N,M}_t \right] &= 
    \frac 1 M \sum_{j} 
    \mathbb E \left[ \left( \exp{(\bar{\mathbf Q}^{N,M,j} \Delta t)} 
    \cdot \begin{bmatrix}
        \mathbf e_{z^{N,M,j}_t}  \\
        0
    \end{bmatrix} \right)_{B+1} \right]
\end{align*}
where $\bar{\mathbf Q}^{N,M,j}$ and $\bar{\mathbf Q}^{M,j}$ are continuous functions of
\begin{align*}
    \lambda^{N,M,j}_t &= \lambda_t \frac M N \sum_{i=1}^N \sum_{k=1}^d \mathbf 1_{x_{t,k}^{i} = j} \mathbf 1_{{u_{t}^{i} = k}},\\
    \lambda^{M,j}_t &= \lambda_t M \int_{\mathcal X \times \mathcal U} \sum_{k=1}^d \mathbf 1_{x_k = j \wedge u = k} \, \mathbb G_t^{M}(\mathrm dx, \mathrm du),
\end{align*}
and as $N \to \infty$, by the conditional law of large numbers \citep[Theorem 3.5]{majerek2005conditional}
\begin{align*}
    \lambda^{N,M,j}_t \to \lambda^{M,j}_t
\end{align*}
a.s. conditional on $z^{N,M,j}_t = z^{M,j}_t = z$ for any $z \in \mathcal Z$. Therefore, again by the continuous mapping theorem, for all $j=1,\ldots,M$ a.s.
\begin{align*}
    \mathbb E \left[ \exp{(\bar{\mathbf Q}^{N,M,j} \Delta t)} \innermid z^{N,M,j}_t = z \right] \to \mathbb E \left[ \exp{(\bar{\mathbf Q}^{M,j} \Delta t)} \innermid z^{M,j}_t = z \right].
\end{align*}

At the same time, $z^{N,M}_t \xrightarrow{d} z^{M}_t$ at all times $t$ as $N \to \infty$ via induction: For $t=0$ trivially $\mathcal L(z^{N,M}_t) = \nu_0 = \mathcal L(z^{M}_t)$. For $t+1$
\begin{align*}
    &\left| \mathbb P(z^{N,M}_{t+1} = z) - \mathbb P(z^{M}_{t+1} = z) \right| \\
    &\leq \sum_{z' \in \mathcal Z} \left| \mathbb P(z^{N,M}_{t} = z') - \mathbb P(z^{M}_{t} = z') \right| \cdot \mathbb P(z^{N,M}_{t+1} = z \mid z^{N,M}_{t} = z') \\
    &\quad + \sum_{z' \in \mathcal Z} \mathbb P(z^{M}_{t} = z') \\
    &\hspace{1cm} \cdot \left| \mathbb P(z^{N,M}_{t+1} = z \mid z^{N,M}_{t} = z') - \mathbb P(z^{M}_{t+1} = z \mid z^{M}_{t} = z') \right|
\end{align*}
where the former tends to zero by induction assumption, while for the latter we have
\begin{multline*}
    \left| \mathbb P(z^{N,M}_{t+1} = z \mid z^{N,M}_{t} = z') - \mathbb P(z^{M}_{t+1} = z \mid z^{M}_{t} = z') \right| \\
    = \left| \prod_{j=1}^M \mathbb E \left[ \left( \exp{(\bar{\mathbf Q}^{N,M,j} \Delta t)} 
    \cdot \begin{bmatrix}
        \mathbf e_{z'^j}  \\
        0
    \end{bmatrix} \right)_{z^j} \innermid z^{N,M}_{t} = z' \right] \right. \\
    - \left. 
    \prod_{j=1}^M \mathbb E \left[ \left( \exp{(\bar{\mathbf Q}^{M,j} \Delta t)} 
    \cdot \begin{bmatrix}
        \mathbf e_{z'^j}  \\
        0
    \end{bmatrix} \right)_{z^j} \innermid z^{M}_{t} = z' \right] \right| \to 0
\end{multline*}
as $N \to \infty$ again as $\bar{\mathbf Q}^{N,M,j} \to \bar{\mathbf Q}^{M,j}$ conditionally a.s. for each $j$.

By Slutzky's theorem (on the conditional probability spaces given $z^{N,M,j}_t = z^{M,j}_t = z$), we have 
\begin{align*}
    &\mathbb E \left[ \left( \exp{(\bar{\mathbf Q}^{N,M,j} \Delta t)} 
    \cdot \begin{bmatrix}
        \mathbf e_{z^{N,M,j}_t}  \\
        0
    \end{bmatrix} \right)_{B+1} \innermid z^{N,M,j}_t = z \right] \\
    &\quad \to \mathbb E \left[ \left( \exp{(\bar{\mathbf Q}^{M,j} \Delta t)} 
    \cdot \begin{bmatrix}
        \mathbf e_{z^{M,j}_t}  \\
        0
    \end{bmatrix} \right)_{B+1} \innermid z^{M,j}_t = z \right]
\end{align*}
for any $z \in \mathcal Z$, such that
\begin{align*}
    &\mathbb E \left[ \left( \exp{(\bar{\mathbf Q}^{N,M,j} \Delta t)} 
    \cdot \begin{bmatrix}
        \mathbf e_{z^{N,M,j}_t}  \\
        0
    \end{bmatrix} \right)_{B+1} \right] \\
    &= \sum_{z \in \mathcal Z} \mathbb E \left[ \left( \exp{(\bar{\mathbf Q}^{N,M,j} \Delta t)} 
    \cdot \begin{bmatrix}
        \mathbf e_{z^{N,M,j}_t}  \\
        0
    \end{bmatrix} \right)_{B+1} \innermid z^{N,M,j}_t = z \right] \\
    &\hspace{2cm}  \cdot \mathbb P \left( z^{N,M,j}_t = z \right) \\
    &\to \sum_{z \in \mathcal Z} \mathbb E \left[ \left( \exp{(\bar{\mathbf Q}^{M,j} \Delta t)} 
    \cdot \begin{bmatrix}
        \mathbf e_{z^{M,j}_t}  \\
        0
    \end{bmatrix} \right)_{B+1} \innermid z^{M,j}_t = z \right] \cdot \mathbb P \left( z^{M,j}_t = z \right) \\
    &= \mathbb E \left[ \left( \exp{(\bar{\mathbf Q}^{M,j} \Delta t)} 
    \cdot \begin{bmatrix}
        \mathbf e_{z^{M,j}_t}  \\
        0
    \end{bmatrix} \right)_{B+1} \right]
\end{align*}
which shows that $\mathbb E \left[ D^{N,M}_t \right] \to \mathbb E \left[ D^{M}_t \right]$ at all times $t$.

Now note that the terms $D_t, D^{M}_t, D^{N,M}_t$ are uniformly bounded by the maximum expected average number of lost packets by dropping all packets, given by the expectation of the Poisson-distributed number of arriving packets $\lambda_t \cdot \Delta t$. Therefore, for any $\varepsilon > 0$ we can choose $T$ such that 
\begin{align*}
    \sum_{t=T}^{\infty} \gamma^t \left( \left| \mathbb E \left[ D_t - D^{M}_t \right] \right| + \left| \mathbb E \left[ D^{M}_t - D^{N,M}_t \right] \right| \right) < \frac \varepsilon 3.
\end{align*}
Consequently choose $M$ sufficiently large such that
\begin{align*}
    \left| \mathbb E \left[ D_t - D^{M}_t \right] \right| < \frac \varepsilon {3T}, \quad \forall t \in \{0,1,\ldots,T-1\}
\end{align*}
and similarly choose $N$ sufficiently large to obtain
\begin{align*}
    \left| \mathbb E \left[ D^{M}_t - D^{N,M}_t \right] \right| < \frac \varepsilon {3T}, \quad \forall t \in \{0,1,\ldots,T-1\}
\end{align*}
according to the prequel, such that $\left| J(\hat \pi) - J^{N,M}(\hat \pi) \right| < \varepsilon$.
\end{proof}
Therefore, our mean-field model is well-motivated for sufficiently large systems, as we will also verify numerically.

\section{Experiments}
\label{sec:experiments}

In this section, we will begin by giving details on the experimental setup. Afterwards, we will demonstrate numerical results of applying reinforcement learning to the MFC MDP problem.

We have $M$ homogeneous queues with exponential service rate $\alpha$ and $N$ clients with Markov modulated arrival rate $\lambda$. 
Beginning with $\lambda_0 \sim \mathrm{Unif}(\{\lambda_h, \lambda_l\})$, at each decision epoch the arrival rate switches between high, $\lambda_h$, and low, $\lambda_l$, levels, using the transition law
\begin{align}
    \mathbb P(\lambda_{t+1} = \lambda_l \mid \lambda_t = \lambda_h) = 0.2, \\
    \mathbb P(\lambda_{t+1} = \lambda_h \mid \lambda_t = \lambda_l) = 0.5.
\end{align}
In general, the experiments could be conducted with more levels of arrival rates and with different modulation rates estimated from a real system, though in our work we will use two arbitrarily chosen values to show the theoretical applicability of our methodology. The values for the system parameters in all of our experiments are given in Table \ref{table:parameters}.

\begin{table}
    \centering
    \caption{System parameters used in the experiments.}
    \label{table:parameters}
    \begin{tabular}{@{}ccc@{}}
    \toprule
    Symbol     & Name          & Value     \\ \midrule
    $\Delta t$          &    Time step size   & $1-10$  \\
    $\alpha$          &    Service rate   & 1   \\
    $(\lambda_h, \lambda_l)$          &    Arrival rates & $(0.9, 0.6)$   \\
    $N$          &    Number of clients   & $1000 - 1000000$   \\
    $M$          &    Number of queues  & $100 - 1000$   \\
    $d$          &    Number of accessible queues   & $2$   \\
    $n$          &    Monte Carlo simulations   & $100$   \\
    $B$          &    Queue buffer size   & $5$   \\
    $\nu_0$          &    Queue starting state distribution   & $[1, 0, 0, \ldots]$   \\
    $D$          &    Drop penalty per job   & $1$   \\ 
    $T$          &    Training episode length   & $500$   \\
    $T_\mathrm{e}$          &    Evaluation episode length   & $50 - 500$   \\ \bottomrule
    \end{tabular}
\end{table}

In order to assess the performance of our MF policy, we compare it to JSQ($d$) and the random policy, RND.
In JSQ($d$), at every decision epoch, $d$ queues are selected out of $M$ and jobs are allocated to the shortest one.
In RND, we similarly select $d$ queues randomly out of $M$ and instead allocate the jobs to a random queue out of the $d$ queues, which will be equivalent to a completely random selection out of $M$ queues for sufficiently large $N \gg M$.
In our work, we have used $d=2$, since in \cite{mitzenmacher2001power} it has been shown that while moving from $d=1$ to $d=2$ shows an exponential increase in performance of JSQ($d$), an additional increase to $d=3$ does not add much in terms of achieved performance.

In order to obtain our MF policy by solving the optimal control problem, we apply proximal policy optimization (PPO) \cite{schulman2017proximal} using the RLlib implementation \cite{liang2018rllib}, a well-known and robust policy gradient reinforcement learning algorithm. The learning algorithm hyperparameters used in our experiments can be found in Table~\ref{table:hyperparameters}.

\begin{table}
    \centering
    \caption{Hyperparameter configuration for PPO.}
    \label{table:hyperparameters}
    \begin{tabular}{@{}ccc@{}}
    \toprule
    Symbol     & Name          & Value     \\ \midrule
    $\gamma$ &   Discount factor &  $0.99$\\
    $\lambda_\mathrm{RL}$ &   GAE lambda &  $1$\\
    $\beta$ &   KL coefficient & $0.2$ \\
    $\epsilon$ &  Clip parameter & $0.3$ \\
    $l_{r}$ &   Learning rate & $0.00005$ \\
    $B_{b}$ &  Training batch size &  $4000$ \\
    $B_{m}$ & SGD Mini batch size &  $128$ \\
    $T_b$ &  Number of epochs & $30$ \\  \bottomrule
    \end{tabular}
\end{table}


\begin{figure}[b]
\center
    \includegraphics[width=0.75\linewidth]{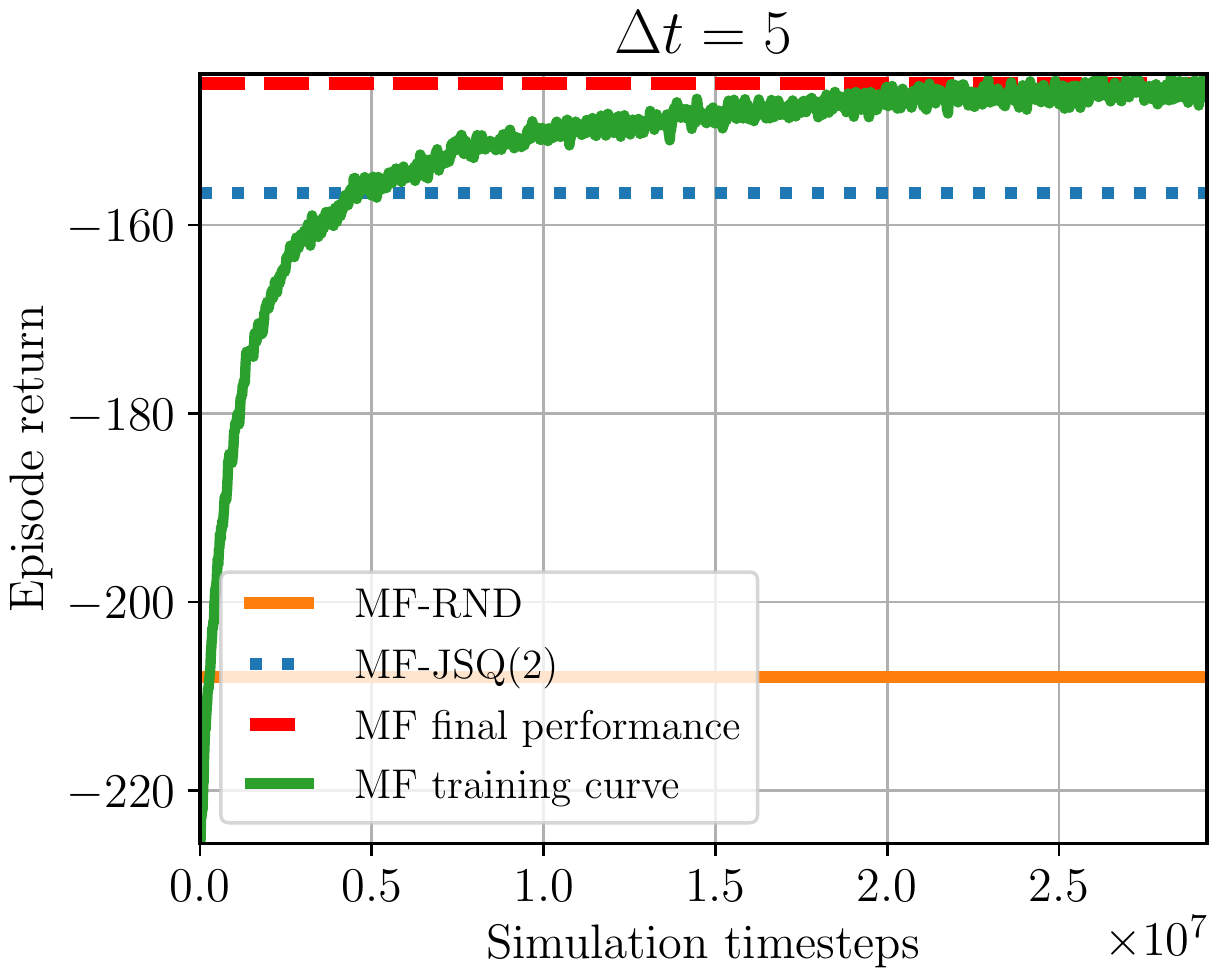}
    \caption{Training curve for the MF policy for $\Delta t=5$ and $T_\mathrm{e} = 500$ timesteps -- i.e. the expected negative number of packet drops per episode during training -- together with a comparison to the MF-JSQ(2) and MF-RND policies. The horizontal lines indicate the estimated expected returns for each policy. The red dotted line indicates the final achieved return of the learned MF policy in the mean-field MDP.}
 \label{fig:progress}
\end{figure}

In Figure~\ref{fig:progress}, we observe the learning curve of the applied reinforcement learning algorithm for $\Delta t=5$ and find that the simple parameterization of the lower-level policies is indeed successful and leads to stable learning. For the demonstrated experiment, we trained in parallel (offline) on $20$ cores of a commodity server CPU for approximately $35$ hours, after which the optimal policy can be applied in practice, to finite systems.
Here, MF-JSQ($2$) and MF-RND refer to the corresponding JSQ and RND policies in the mean-field model, i.e. each applies a fixed $h_t$ regardless of the current queue state distribution $\nu_t$. In the case of MF-JSQ given by
\begin{align}
    h_t(u \mid \bar z) = 
    \begin{cases}
        0 & \quad \text{ if } u \centernot\in \argmin_{u'} \bar z_{u'} \\
        \frac 1 {N_\mathrm{min}} & \quad \text{ else }
    \end{cases}
\end{align}
where $N_\mathrm{min}$ is the number of actions $u$ that minimize  the chosen queue's state $\bar z_{u}$. In the case of MF-RND, we similarly choose
\begin{align}
    h_t(u \mid \bar z) = \frac 1 {|\mathcal U|}, \quad \forall (\bar z, u) \in \mathcal Z^d \times \mathcal U \, 
\end{align}

As expected, indicated by the horizontal lines, the JSQ($2$) and random (RND) assignment policies in the mean-field case are both suboptimal for the chosen delay time of $\Delta t = 5$, and our reinforcement learning approach is capable of finding better load balancing policies after approximately $5$ million simulated decision epochs. Though we have tried Dirichlet-parameterized upper-level policies to directly output simplex-valued actions in order to eliminate the need for manual normalization, we found that performance was significantly worse, hence motivating our approach. 

\begin{algorithm}[t]
    \caption{Application of MFC policy in finite system}
    \label{alg}
    \begin{algorithmic}[1]
        \STATE \textbf{Input}: System parameters from Table~\ref{table:parameters}
        \STATE \textbf{Input}: Markovian upper-level policy $\tilde \pi = \{ \tilde \pi_t \}_{t \geq 0}$
        \STATE Initialize $\lambda_0 \sim \mathrm{Unif}(\{\lambda_h, \lambda_l\})$.
        \FOR {$j = 1, \ldots, M$}
            \STATE Initialize queue states $z_0^{j} \sim \nu_0$.
        \ENDFOR
        \FOR {$t = 0, 1, \ldots, T_\mathrm e$}
            \STATE Compute empirical distribution $\mathbb H_t^{M} = \frac{1}{M} \sum_{j=1}^M \delta_{z_t^{j}}$.
            \STATE Sample decision rule $h_t \sim \tilde \pi_t(\mathbb H_t^M, \lambda_t)$.
            \FOR {$i = 1, \ldots, N$}
                \STATE Sample agent state $x_t^{i} \sim \otimes_{k=1}^d \mathrm{Unif}(\{1,\ldots,M\})$.
                \STATE Compute anonymous state $\bar z_t^{i} = (z_t^{x_{t,1}^{i}}, \ldots, z_t^{x_{t,d}^{i}})$.
                \STATE Sample agent action $u_t^{i} \sim h_t(\bar z_t^{i})$.
            \ENDFOR
            \FOR {$j = 1, \ldots, M$}
                \STATE Simulate continuous-time Markov chain $y^{j}$ with jump rates $\lambda_t^j, \alpha$ and $y^{j}(0) = z_t^j$ for $\Delta t$ time units.
                \STATE Count number of dropped packets.
                \STATE Set queue state $z_{t+1}^j = y^j(\Delta t)$.
            \ENDFOR
            \STATE Sample $\lambda_{t+1} \sim \mathbb P(\lambda_{t+1} \mid \lambda_t)$.
        \ENDFOR
        \RETURN Number of dropped packets.
    \end{algorithmic}
\end{algorithm}

\begin{figure*}
\centering
\includegraphics[width=0.99\linewidth]{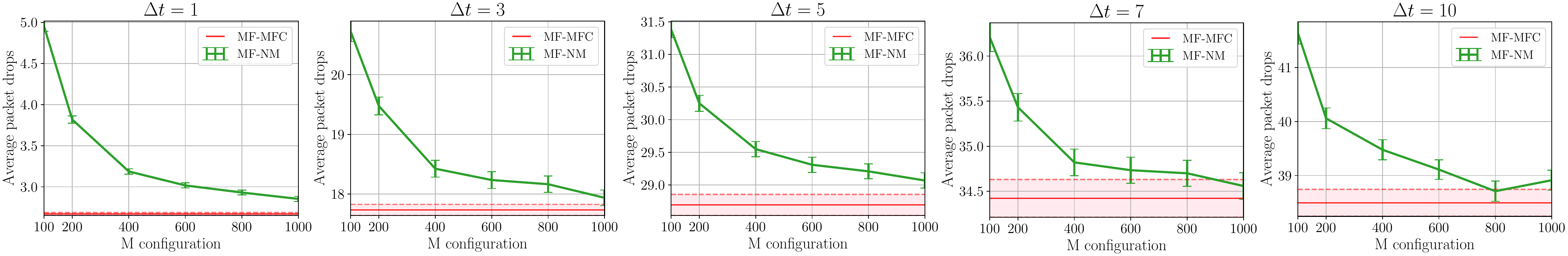}
    \caption{Comparison of the estimated expected packet drops (lower is better) of MF policies over the number of queues $M$ in the finite system for different values of $\Delta t$, together with $95\%$ confidence intervals depicted as shaded regions and error bars. Here, we use total running times of approximately $500$ time units, and $N = M^2$ to fulfill $N \gg M$. The red dotted line indicates the equivalent achieved return of the learned MF policy in the mean-field control MDP, i.e. the limiting model as $N \gg M \to \infty $. It can be observed that as the system size $N = M^2$ increases, the performance under the MF policy (green) becomes increasingly close to the mean-field system performance (red), validating the accuracy of our mean-field formulation.}
 \label{fig:all_delta_t}
\end{figure*}


\paragraph{Performance comparison on finite systems}
 We will now compare the performance of the evaluated load balancing algorithms on systems of finite size. For simulation of the finite-agent and finite-queue system, we simulate the continuous-time Markov processes exactly by sampling exponential waiting times for all events according to the Gillespie algorithm \cite{gillespie1977exact}. For an easy comparison between different $\Delta t$, we set the episode lengths $T_\mathrm e$ for evaluation to the integer nearest to $\frac{500}{\Delta t}$. Pseudocode for simulating and applying our MF policy in the finite system is given in Algorithm~\ref{alg} \footnote{\url{https://github.com/AnamTahir7/mfc_large_queueing_systems.git}}. 

In Figure~\ref{fig:all_delta_t}, we show that the performance of the final learned MF policies over a wide range of delays $\Delta t$ and system sizes $(N,M)$.
It can be seen that the overall achievable performance of our MF policy increases up to the performance achieved in the MFC MDP (red dotted line) as the system size $(N,M)$ becomes sufficiently large ($N \gg M \gg 1$). 
Hence, our findings empirically validate the fact that our mean-field approximations are indeed accurate for sufficiently large system sizes.


\begin{figure}[b]
\center
\includegraphics[width=0.95\linewidth]{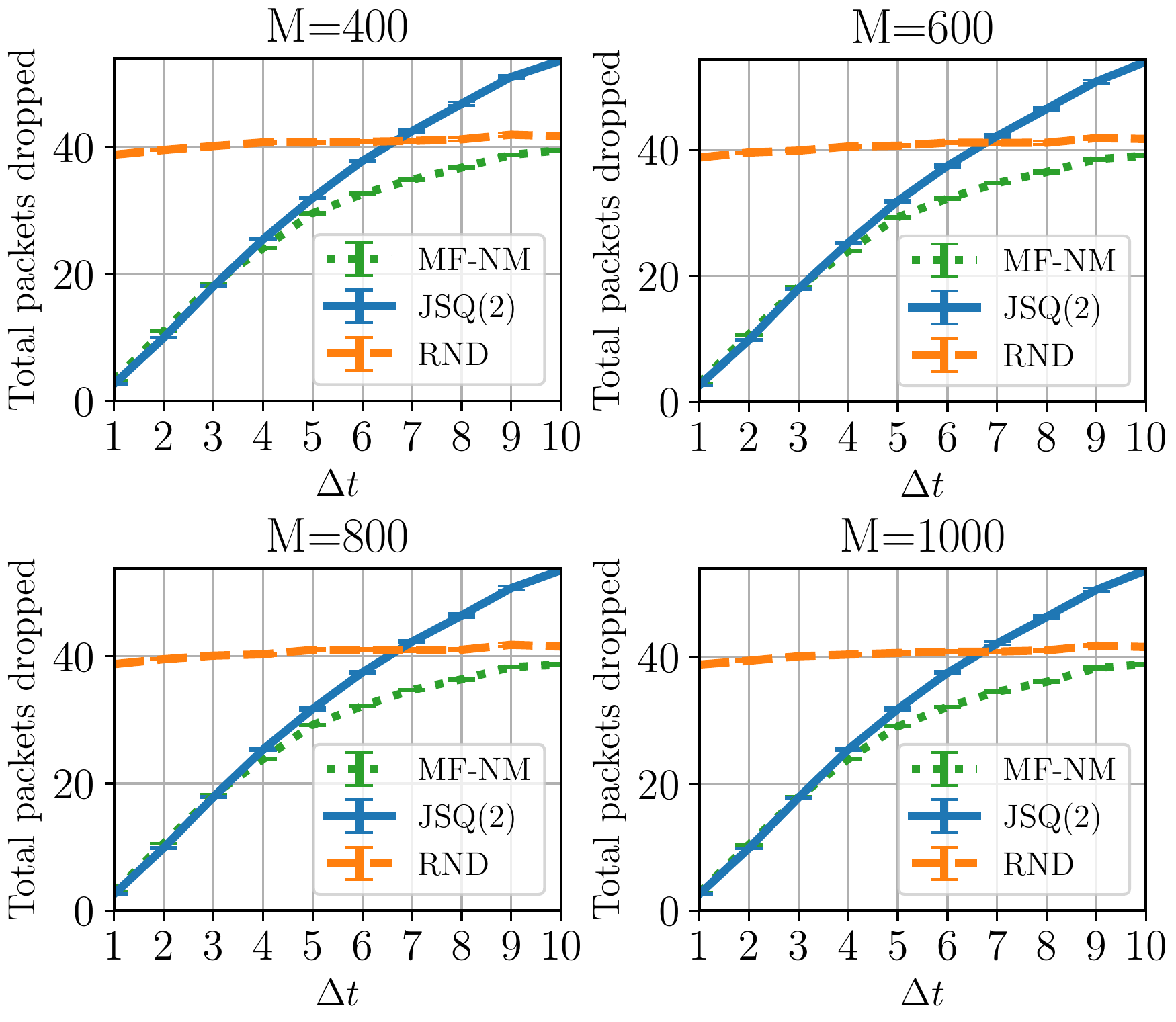}
    \caption{Comparison of the estimated expected packet drops of MF, JSQ($2$), RND policies together with $95\%$ confidence intervals for different configurations of $M$ and $N = M^2$. We keep the total running time of each setting approximately equal to $500$ time units to compare the effect of $\Delta t$. It can be observed that as $\Delta t$ rises, the achievable performance by choosing emptier queues degrades.}
 \label{fig:equal_time}
\end{figure}


The returns for the policies at each $\Delta t$, for the case where all experiments are run for approximately equal overall time instead of an equal number of decision epochs, are given in Figure~\ref{fig:equal_time}.
Here, we have trained a separate MF policy for each of the $\Delta t$ and compared to JSQ($2$) and RND. It can be seen that -- as expected due to fewer updates -- the overall achievable performance in the system worsens as the synchronization delay $\Delta t$ of the system increases. It can be seen that MF achieves better performance than JSQ($2$) starting from $\Delta t > 2$, while it outperforms RND in all cases. This stems from the fact that reinforcement learning only finds approximately optimal solutions. Nonetheless, at an intermediate level of synchronization delay beginning with $\Delta t = 3$, our learning-based methodology appears to be able to find a better policy than the optimal policies for $\Delta t \to 0$ (JSQ($2$)) and $\Delta t \to \infty$ (RND). Even for small $\Delta t=1$, our MF policy has comparable performance to the optimal JSQ($2$) policy, as long as $N,M$ are sufficiently large. As $\Delta t$ keeps increasing, MF and RND are therefore expected to perform equally good in sufficiently large systems as long as we indeed have $N \gg M$.


\begin{figure}
\center
\includegraphics[width=0.95\linewidth]{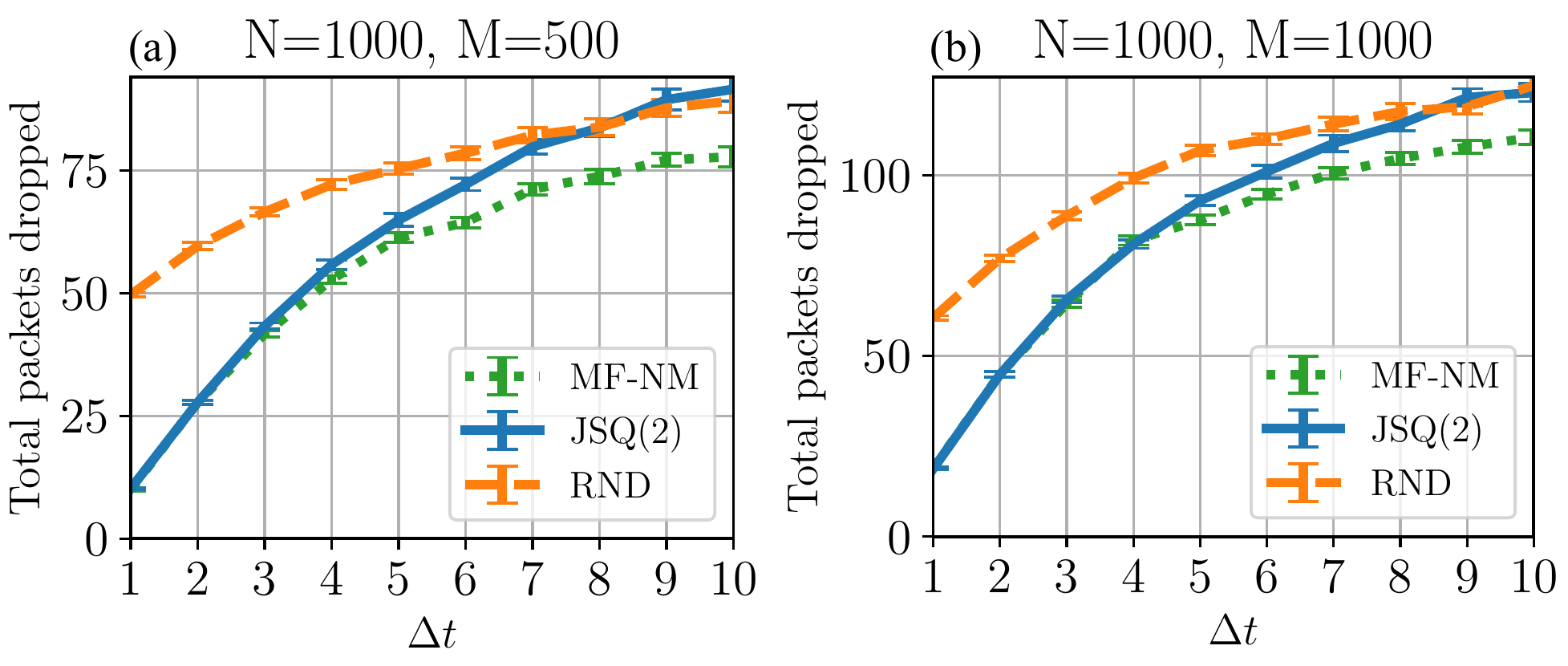}
    \caption{Comparison of the estimated expected packet drops of MF, JSQ($2$), RND policies together with $95\%$ confidence intervals for the same setting as in Figure~\ref{fig:equal_time}, equal total running time, for the case when  $M=1000$, $N=\frac{M}{2}$ and t$N=M$. As $\Delta t$ increases, the performance of our MF policy performs better than the other policies, even when  $N \not\gg M$.}
 \label{fig:equal_time_ablation}
\end{figure}


Finally, we perform experiments for $N  \not\gg M$, i.e. we violate the formal approximation assumption used to obtain our mean-field system. Even though the assumptions made in our approximation are violated, our policy nonetheless obtains good comparative performance. As shown in Figure~\ref{fig:equal_time_ablation}, we find that the qualitative performance differences remain the same for around $1000$ agents and queues. It can also be observed that the random policy no longer obtains approximately equal performance over $\Delta t$, which is caused by the fact that the queues are increasingly sampled unequally often by an agent, and resampling resolves the resulting increased focus on a subset of queues.

\section{Discussion}

In this work, we have proposed a mean-field-control-style formulation, with enlarged state-action space, for large-scale distributed queuing systems with synchronization delays. We have achieved this by formulating the finite-agent finite-queue system and considering $N \to \infty$, $M \to \infty$.  

Firstly, we provide theoretical performance guarantees which show that the performance in the $N,M$ system becomes arbitrarily close to the performance in the MFC system as long as $N,M$ are large enough.
Then, assuming a synchronous system with exact discretization of the underlying processes, we end up with an exactly discretized discrete-time Markov decision process on which we have applied reinforcement learning algorithms. As a result, we find that our learned solution can outperform the delay-free-optimal JSQ($d$) policy as well as the infinite-delay-optimal random policy in the regime of intermediate delays $\Delta t$, even if $N \centernot \gg M$ as long as the system size $N,M$ is sufficiently large.

An interesting future direction could be further extensions to the model such as non-exponential inter-arrival and service times, partial observability as well as explicitly modelling the case where $N$ is not significantly larger than $M$. To allow for better scaling of the reinforcement learning algorithm to very large queue sizes, it may be of interest to apply further limiting, real-valued approximations of the queue states as $B \gg 1$. One straightforward extension would be to used heterogenous service rates. Finally, an implementation of the developed methods in a real world system may be of interest. We hope that our work inspires further work at the intersection of mean-field control theory and distributed queuing systems.

\begin{acks}
This work has been co-funded by the German Research Foundation (DFG) as part of sub-project C3 within the Collaborative Research Center (CRC) 1053 – MAKI and the LOEWE initiative (Hesse, Germany) within the emergenCITY center.
\end{acks}


\bibliographystyle{ACM-Reference-Format}
\bibliography{references}

\end{document}